\documentclass[%
 reprint,
 amsmath,amssymb,
 aps,
 pra,
]{revtex4-2}

\usepackage{xcolor}
\usepackage{amsmath}
\usepackage{amsthm} 
\usepackage{pgfplots}
\usepackage[hidelinks=true]{hyperref}
\usepackage{xurl}
\usepackage{tikz}
\usetikzlibrary{quantikz2}
\usepackage{pgfplotstable} 
\pgfplotsset{compat=1.18}
\usepgfplotslibrary{groupplots} 

\newtheorem{theorem}{Theorem}[section]
\newtheorem{lema}[theorem]{Lemma}

\newcommand{\ketbra}[1]{\left|#1\right>\left<#1\right|}
\usepackage{graphicx}
\usepackage{dcolumn}
\usepackage{bm}
\usepackage{cleveref}

\begin{document}

\preprint{APS/123-QED}

\title{Cloning Encrypted Quantum States in Arbitrary Dimensions}

\author{Filip-Ioan Ceară}
\email{filipceara@gmail.com}
\affiliation{%
 Advanced Technologies Institute, 10 Dinu Vintilă, Bucharest, Romania
}%




\date{\today}

\begin{abstract}
Recently, Yamaguchi and Kempf  
proved that encrypted qubits can be cloned. In this work, we generalize the encrypted cloning protocol and prove that it also applies to higher-order quantum systems. Given that a straightforward generalization of the protocol using the exponential of the shift and phase operators fails to satisfy the unitary requirement for a quantum gate, we propose a different approach. We introduce a new operator to be used in the encryption process and show that it is unitary. We adapt the decryption operator from the reference paper to fit in the framework of multi-level quantum systems. We analyze the circuit implementation of the proposed operators and show that the overhead imposed by larger dimensions scales linearly with qudit dimension.
\end{abstract}

\maketitle


\section{\label{sec:introduction}Introduction}
Generalizing qubit algorithms or protocols to qudits is a challenging task in quantum information theory~[\onlinecite{Wilde_2013}]. The main advantage of multi-dimensional quantum states can especially be observed in the field of quantum communication~[\onlinecite{doi:10.1126/sciadv.1701491}].
A qudit can implement more efficient quantum communication protocols because it has an inherently larger alphabet size and is more robust to noise~[\onlinecite{Ionicioiu2016}]. This increased ``noise threshold'' makes qudit-based systems more suitable for long-distance communication
Furthermore, multi-qudit systems allow more complex entanglement structures, thus enabling more sophisticated secret-sharing schemes~[\onlinecite{PhysRevA.92.030302}].

Recently, Yamaguchi and Kempf~[\onlinecite{y4y1-1ll6}] presented a protocol that allows the cloning of encrypted states of a qubit, without violating the no-cloning theorem. We solve one of the open problems stated in their paper, showing that the cloning of encrypted states is not restricted to qubit-based systems. Thus, we establish that this protocol is intrinsic to quantum information encoded in finite dimensions. This result highlights the protocol's scalability and relevance. Thus, one can leverage the protocol's scalability and the advantages of qudit systems to propose new schemes in quantum cryptography.

In this work, we discuss the fact that the natural generalization of the encryption operator, as the exponential of a generalized Pauli operator, cannot be performed for \(d\geq 3\). Thus, we propose a generalization of the encryption operator based on a constant-amplitude zero-autocorrelation (CAZAC) sequence. For the decryption matrix, we generalize the form given in~[\onlinecite{y4y1-1ll6}] to the qudit case. We show that both of the proposed matrices are unitary. We demonstrate that, within the protocol, the proposed matrices meet the imposed requirements. Thus, the encryption operator creates a state indistinguishable from a maximally mixed state, while the decryption operator successfully retrieves the information originally stored in the qudit of interest, as detailed in the protocol. 

The structure of the article is the following.
In Section \ref{sec:Preliminaries}, we define the problem statement and introduce the notations that will be used throughout the paper. In Section \ref{sec:analysis}, we perform the analysis of the proposed matrices for encryption and decryption operations. The methods for implementing the unitary matrices are analyzed in Section \ref{sec:impl}, where we also evaluate the gate complexity and assess the number of one- and two-qudit gates required to implement the protocol. Finally, we summarize the work in Section \ref{sec:concl}. All the proofs of the statements that we have used during this work are detailed in the appendices.

\section{Preliminaries}

\label{sec:Preliminaries}
A qudit represents a multi-dimensional quantum state that generalizes the qubit (which can be represented mathematically as a vector in a two-dimensional Hilbert space). Let \(\mathcal{H}_d\) be a Hilbert space with dimension \(d = \dim\mathcal{H}_d\). Thus, any qudit \(\ket{\psi} \in \mathcal{H}_d\) can be written as \(\ket{\psi} = \sum_{i=0}^{d-1} \alpha_i\ket{i}\), using the generalized computational basis \(\{\ket{0},\ket{1},\dots,\ket{d-1}\}\)~[\onlinecite{Weyl1950}].

We define the generalized Pauli operators, often called Weyl operators, in the computational basis as:
\begin{eqnarray}
    \label{X_operator}  
    X_d\ket{k} = \ket{k \oplus_d 1}\\
    \notag\\
    \label{Z_operator}
    Z_d\ket{k} = \omega^{k}\ket{k},
\end{eqnarray}
where \(\omega = e^{\frac{2\pi i}{d}}\) is a root of unity of order \(d\) and \(\oplus_d\) represents the addition modulo \(d\). Respectively, the notation \(\ominus_d\) is equivalent to subtraction modulo \(d\). The \(X_d\) is also called the shift operator because it acts as a permutation in the computational basis, while the \(Z_d\) is often named the phase operator. It is clear that \(X_d^d = Z_d^d = I_d\), where \(I_d\) represents the identity matrix of dimension \(d\). It can be easily shown that \(X_d^{\dagger} = X_d^{-1}\) and \(Z_d^{\dagger} = Z_d^{-1}\). Thus, unlike the case of \(d=2\), where \(X^{-1} = X\), and \(Z^{-1}=Z\), for any \(d\geq 3\) the generalized Pauli operators are no longer hermitian. 

The Fourier transform \(F\) maps the two generalized Pauli operators
\begin{equation}
\label{XZ_TRANS}
    X_d = F^{\dagger}Z_dF,
\end{equation}
and is defined in~[\onlinecite{Nielsen_Chuang_2010}] as
\begin{equation}
    F\ket{k} = \frac{1}{\sqrt{d}}\sum_{j=0}^{d-1}\omega^{jk}\ket{j}.
\end{equation}

For an arbitrary one-qudit gate \(U\), we can define two types of controlled gates. The first one is defined as the \(C(U)\) gate~[\onlinecite{Ionicioiu2016}], acting on the control qudit \(\ket{j}\) and the target qudit \(\ket{k}\), as
\begin{equation}
    \label{eq:def_CU}
    C(U) \ket{j}\ket{k} = \ket{j}U^{j}\ket{k}.
\end{equation}
The second is the $p$-controlled gate, which is denoted as \(C_p(U)\) and is acting as
\begin{equation}
    \label{k_controlled}
    C_p(U)\ket{j}\ket{k} = \ket{j}U^{j\delta_{p,j}}\ket{k}.
\end{equation}
In this case, the quantum gate \(U\) is applied to the target qudit only when \(j=p\), as specified by the Kronecker-delta function \(\delta_{j,p}\).
Using the definition in \eqref{eq:def_CU}, the generalized controlled Pauli operators \(C(X_d)\) and \(C(Z_d)\) are
\begin{eqnarray}
    \label{def_CXd}
    C(X_d)\ket{j}\ket{k} = \ket{j}\ket{k \oplus_d j},\\\notag\\
    C(Z_d)\ket{j}\ket{k} = \omega^{jk}\ket{j}\ket{k}.
\end{eqnarray}

The quantum state
\begin{equation}
    \label{Gen_Bell}
    \ket{\Phi_d} = \frac{1}{\sqrt{d}}\sum_{p=0}^{d-1}\ket{p}\ket{p}
\end{equation}
represents the generalized Bell state and could be easily obtained by applying the \(F\) and \(C(X_d)\) gates, as 
\begin{equation}
    C(X_d)\cdot(F\otimes I_d)\ket{00} = \ket{\Phi_d}.
\end{equation}

The generalization that we are proposing for the encryption cloning protocol, which is developed in~[\onlinecite{y4y1-1ll6}], assumes the existence of a data qudit \(\ket{\psi}_A\) and \(n\) pairs of maximally entangled states \(\ket{\psi}_{S_i,N_i} = \ket{\Phi_d}\). The \(S_i\) qudits represent those used to clone the encrypted state, while the \(N_i\) qudits are stored locally and used to retrieve the initial information stored in the data qudit. In the encryption process, a unitary gate will be applied to the data qudit \(A\) and to all the \(S_i\) qudits. The application of the encryption matrix should produce a quantum state that prevents any measurement of any of the \(S_i\) qudits from revealing information about the data qudit. Thus, the qudits \(S_i\) could be shared among \(n\) parties, while keeping all \(N_i\) qudits locally, where the data qudit originated. After choosing one of the qudits \(S_i\), applying the decryption operator to that qudit and all of the \(N_i\) qudits allows a single party to have access to the original quantum state stored in the data qudit, while all of the remaining \(n-1\) parties obtain only a maximally mixed state.

\section{\label{sec:analysis}Operator definitions}
In this section, we provide definitions of the encryption and decryption operators and prove that both are indeed unitary matrices. Moreover, we show that both matrices satisfy the functional requirements stated in Section \ref{sec:Preliminaries}.
\subsection{Encryption operator}
Similar to~[\onlinecite{y4y1-1ll6}], we define the operators \(P_X\) and \(P_Z\) 
\begin{eqnarray}
    \label{PX_operator}
    P_X = X_d^{\scriptscriptstyle(A)} \otimes X_d^{\scriptscriptstyle(S_1)} \otimes \dots \otimes X_d^{\scriptscriptstyle(S_n)}\\
    \notag\\
    \label{PZ_operator}
    P_Z = Z_d^{\scriptscriptstyle(A)} \otimes Z_d^{\scriptscriptstyle(S_1)} \otimes \dots \otimes Z_d^{\scriptscriptstyle(S_n)},
\end{eqnarray}
where the \(X_d^{(i)}\), and \(Z_d^{(i)}\) are the generalized Pauli operators described in \eqref{X_operator}, and \eqref{Z_operator} respectively. We use the notation \(X_d^{(i)}\) to illustrate the application of the \(X_d\) to the qudit labeled with \(i\). We express the encryption gate as in the following expression:
\begin{equation}
    \label{u_enc_operator}
    U_{enc}^{(n)} = V(P_X)V(P_Z),
\end{equation}
using a generic map \(V\). Therefore, we define the required properties of the operator \(V(P)\), such that the encrypted state in each of the \(\ket{\psi}_{S_i}\) reveals no information about the original state of the data qudit. The natural generalization of the operator \(V(P)\) is the one described in the following equation 
\begin{equation}
    \label{exponential}
    V(P) = e^{-i\theta P},
\end{equation}
which, for the case of \(\theta = \frac{\pi}{4}\), results in the same form of the operators as in the original paper~[\onlinecite{y4y1-1ll6}]. For the case of \(d = 2\) the matrices \(X_2\) and \(Z_2\) are hermitian matrices, resulting that the operator in \eqref{exponential} can be written as
\begin{equation}
    \label{exponential_d2}
    V(P) = \cos{\theta}\cdot I_2 -i\sin{\theta}
    \cdot P.
\end{equation}
For the case of \(d\geq 3\), due to the property of \(X_d\) and \(Z_d\) not being hermitian matrices, the operator described in \eqref{exponential} yields a non-unitary matrix~[\onlinecite{Nielsen_Chuang_2010}]. Thus, this natural generalization cannot be used for qudits.

To achieve the condition that the operator is unitary, we propose the operator
\begin{equation}
    \label{new_operator}
    V(P) =\frac{1}{\sqrt{d}}\sum_{k=0}^{d-1}c(k)P^{k},
\end{equation}
where the \(c(k)\) represents a Chu Sequence, a particular case of a Zadoff-Chu Sequence~[\onlinecite{1054840}]. All Zadoff-Chu sequences are CAZAC, meaning they have a perfectly flat power spectrum and their periodic autocorrelation is a Kronecker-delta function. The general expression for this type of sequence is given in the following equation
\begin{equation}
    \label{zadoff-chu}
    zc(k) = \exp\Big({-i\frac{\pi uk(k+c_f+2q)}{d}}\Big),
\end{equation}
with the requirement that \(c_f = d\pmod2\), and \(gcd(u, d)=1\). The \(c(k)\) sequence is obtained by choosing the integer parameters \(u=1\) and \(q=0\). Thus, we obtain the following operator that will be used in the encryption gate
\begin{equation}
    \label{final_operator}
    V(P) = \frac{1}{\sqrt{d}}\sum_{k=0}^{d-1}e^{-i\frac{\pi k(k+d\%2)}{d}}P^k,
\end{equation}
where we have defined \(\%\) as the remainder of division by 2.
Using the expression of the the operator \(V(P)\) as expressed in \eqref{final_operator}, with  \(P_X\), and \(P_Z\), as described in \eqref{PX_operator} and \eqref{PZ_operator}, respectively, the encryption matrix has the following expression
\begin{equation}
    \label{definition_enc_mine}
    U_{enc_{d}}^{(n)} = \frac{1}{d}\sum_{k=0}^{d-1}e^{-i\frac{\pi k(k+d\%2)}{d}}{(P_X)}^{k}\sum_{l=0}^{d-1}e^{-i\frac{\pi l(l+d\%2)}{d}}{(P_Z)}^{l}.
\end{equation}
Following the form of the encryption gate in \eqref{definition_enc_mine}, expressing the operators \(P_X, P_Z\) by indicating the qudit they are being applied to, and using the notation \(c_{kl} = c(k)\cdot c(l) = e^{-i\frac{\pi k(k+d\%2)}{d}}\cdot e^{-i\frac{\pi l(l+d\%2)}{d}}\), we obtain the following expression
\begin{widetext}
\begin{equation}
    \label{definition_enc_mine_final}
    U_{enc_{d}}^{(n)} = \frac{1}{d}\sum_{k,l=0}^{d-1}c_{kl} ({X_{ d}^{(A)}})^k({Z_{ d}^{(A)}})^{l} \otimes \bigotimes_{i=1}^{n}({X_d^{(S_i)}})^{k}({Z_d^{(S_i)}})^{l}.
\end{equation}
\end{widetext}
It can be easily seen that, for each of the coefficients \(c_{kl}\), \(|c_{kl}|=1\), thus resulting in an equal superposition of operators. One property of the Zadoff-Chu sequences is that they exhibit orthogonality between cyclically shifted versions of themselves~[\onlinecite{1054840}]. Thus, each sequence of coefficients \(c(k)\) exhibits similar properties to those of white noise, as the autocorrelation is a Kronecker-delta function. It has been shown that the Kronecker product of two CAZAC sequences is also a CAZAC sequence~[\onlinecite{144727}]. Thus, when using the coefficients \(c_{kl}\), we are obtaining a sequence whose autocorrelation is equal to \(0\) everywhere except the origin. We can observe the behavior in \Cref{fig:2d_ZC}, and conclude that the coefficients that we use are similar to samples extracted from a pseudorandom generator, and can successfully be used for the encryption operator. 
\begin{figure}[!b]
    \centering
    \includegraphics[width=0.95\linewidth]{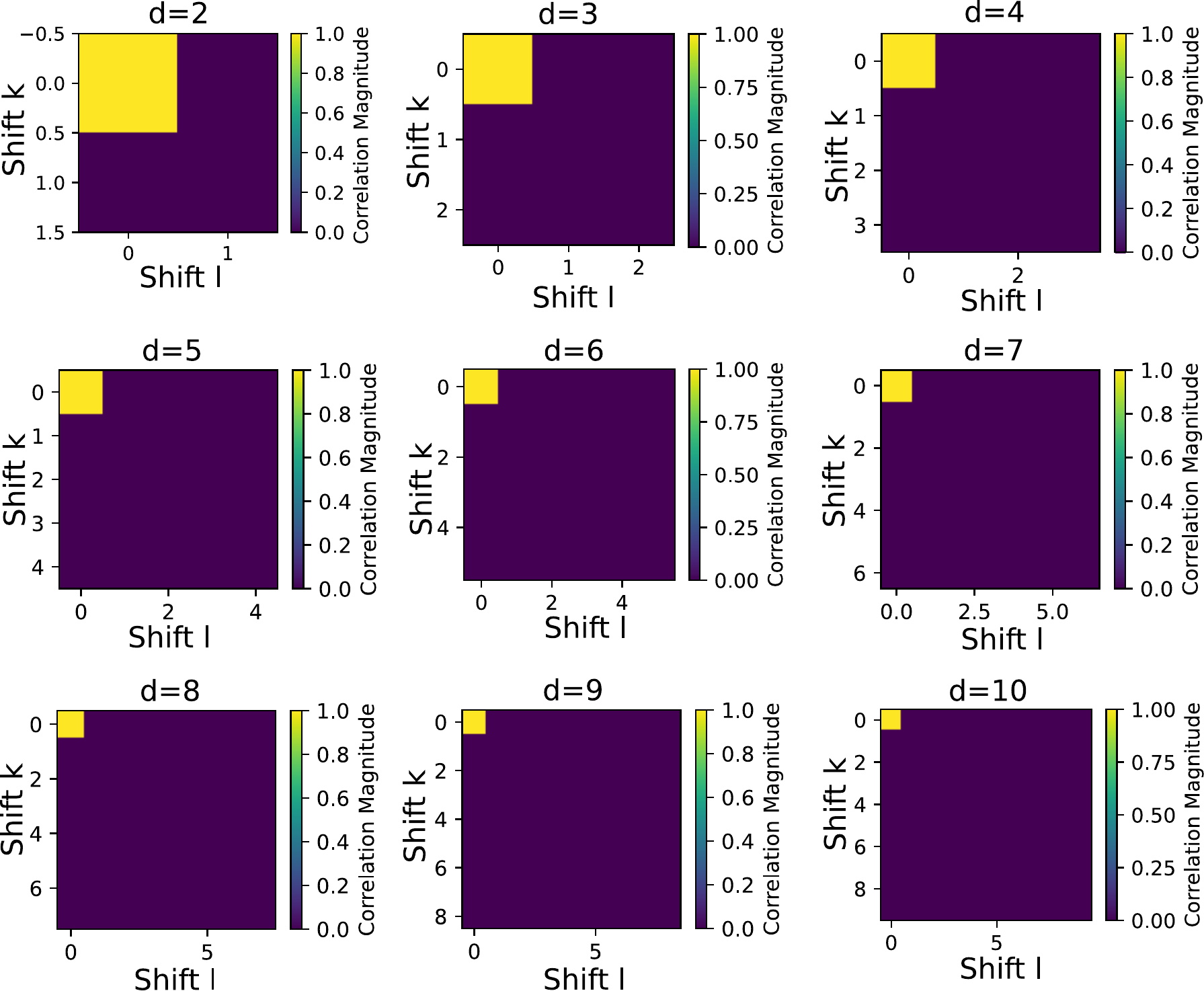}
    \caption{2D Autocorrelation for \(c_{kl}\) coefficients, for different values of the qudit dimension. All the plots show a maximum value of \(1\) at \((0,0)\) and the value \(0\) for all the other combinations.}
    \label{fig:2d_ZC}
\end{figure}

We further show that the proposed operator from \eqref{definition_enc_mine_final} is a consistent generalization because it produces the same expression as that given in the original paper. Thus, for the case of \(d=2\), the coefficients \(c_{kl}\) have the following values \(c_{00} = 1, c_{01}=-i, c_{10}=-i, c_{11}=-1\), and thus the encryption unitary is given as
\begin{align*}
U_{enc_{2}}^{(n)} = \frac{1}{2} \Big(I_{2}^{\scriptscriptstyle(A)} \otimes \bigotimes_{i=1}^{n}I_2^{\scriptscriptstyle(S_i)} + (-i)\cdot Z_{2}^{\scriptscriptstyle(A)}\otimes \bigotimes_{i=1}^{n}Z_2^{\scriptscriptstyle(S_i)} \\
+ (-i)\cdot X_{2}^{\scriptscriptstyle(A)}\otimes \bigotimes_{i=1}^{n}X_2^{\scriptscriptstyle(S_i)} + \\ (-1)(-i)^{n+1}\cdot Y_{2}^{\scriptscriptstyle(A)}\otimes \bigotimes_{i=1}^{n}Y_2^{\scriptscriptstyle(S_i)} \Big),
\end{align*}
resulting in the same operator as the one given in \cite{y4y1-1ll6}. 

Recalling from \eqref{definition_enc_mine}, the encryption operator can be written as the product of two matrices \(A_X\) and \(A_Z\), as shown
\begin{align*}
{U_{enc_d}^{(n)}} &= \frac{1}{\sqrt{d}}\sum_{k=0}^{d-1}e^{-i\frac{\pi k(k+d\%2)}{d}}P_X^{k}\frac{1}{\sqrt{d}}\sum_{l=0}^{d-1}e^{-i\frac{\pi l(l+d\%2)}{d}}P_Z^{l}\\
&=A_X \cdot A_Z.
\end{align*}
To prove that \(U_{enc_d}^{(n)}\) is unitary, we show that each of the matrices \(A_X\) and \(A_Z\) is unitary. Each of the two matrices is the sum of \(d-1\) unitary terms; thus, it is not trivial that the matrix is unitary. The proof that each of the matrices \(A_X\) and \(A_Z\) is unitary is given in Appendix \ref{unitary_appendix_enc} and is based on the orthogonality of terms in the Zadoff-Chu sequence. As a consequence, the matrix could be implemented as a gate and applied to a quantum circuit.

\subsection{Decryption operator}
For the decryption matrix, we follow the construction mechanism from~[\onlinecite{y4y1-1ll6}], but adapt it for any dimension \(d\geq3\). In the original paper, the authors used the following result
\begin{equation}
    \Big(\sigma_\mu^{(N_i)}{}^{\mathsf{T}}\otimes \sigma_\mu^{(S_i)}\Big) \ket{\Phi_2} = \ket{\Phi_2}, \forall\mu\in \{0,\dots,3\},
\end{equation}
where the operators \(\sigma_\mu\) represent the Pauli operators \(\{I_2, X_2, Y_2, Z_2\}\). For \(d\geq 3 \), this result does not apply; thus, we propose the following theorem, whose proof is provided in Appendix \ref{proof_th2}.
\begin{theorem}
\label{th_2}
    The following identity
    \[
    \Big({X_d^{\scriptscriptstyle(Q_1)}}^{k_1}{Z_d^{\scriptscriptstyle(Q_1)}}^{-k_2} \otimes {X_d^{\scriptscriptstyle(Q_2)}}^{k_1}{Z_d^{\scriptscriptstyle(Q_2)}}^{k_2}\Big)\ket{\Phi_d}=\ket{\Phi_d}
    \]
    holds for any \(d\geq 2\) and \(k_1,k_2 \in \{0\dots d-1\}\), where \(Q_1, Q_2\) represent the two qudits involved.
\end{theorem}

Taking into account \Cref{th_2} and using the notations \(\Big({X_d^{\scriptscriptstyle(S_1)}}^{k}{Z_d^{\scriptscriptstyle(S_1)}}^{l}\otimes {I_d}^{\scriptscriptstyle(N_1)}\Big) \equiv O_{kl}^{\scriptscriptstyle(S_1, N_1)}\) and \(( I_d^{\scriptscriptstyle (S_1)} \otimes {F^{\scriptscriptstyle (N_1)}}^2) \left( \textstyle \sum_{c=0}^{d-1} \left(X_d^{\scriptscriptstyle (S_1)}\right)^{2c}  \otimes (|c\rangle\langle c|)^{\scriptscriptstyle (N_1)} \right) \equiv C^{\scriptscriptstyle (S_1, N_1)}\), we propose the following decryption operator:
\begin{widetext}
\begin{equation}
\label{U_dec_mine}
    U_{dec_{d}}^{(n)} = \sum_{k,l=0}^{d-1}c_{kl}^{-1}\text{SWAP}^{\scriptscriptstyle(S_1,N_1)}C^{\scriptscriptstyle (S_1,N_1)}
    O_{kl}^{\scriptscriptstyle(S_1N_1)}\ketbra{\Phi_d}{O_{kl}^{\scriptscriptstyle(S_1N_1)}}^\dagger \\ 
    \otimes \bigotimes_{j=2}^{n}\Big({X_d^{\scriptscriptstyle(N_j)}}^{k}{Z_d^{\scriptscriptstyle(N_j)}}^{-l}\Big).
\end{equation}
\end{widetext}
Hence, we defined the decryption operator that retrieves the information stored originally in the \(\ket{\psi}_A\) to the \(\ket{\psi}_{S_1} \). Without loss of generality, the construction can be easily modified so that the information retrieval can be made on any other \(\ket{\psi}_{S_i}\). This transformation is possible because the state obtained after the encryption process is symmetrical, in the sense that the same operators have been applied to all the  \(\ket{\psi}_{S_i}\). The construction that we proposed for the decryption operator is equivalent to a SWAP gate between the data qudit \(\ket{\psi}_A\) and \(\ket{\psi}_{N_1}\), as we have proved in Appendix \ref{th_swap}. Unlike the expression in~[\onlinecite{y4y1-1ll6}], multi-dimensional quantum states require an additional SWAP gate and a gate $C$, composed of two one-qudit and two two-qudit gates.

As the expression given for \(U_{dec_d}^{(n)}\) in \eqref{U_dec_mine} is a sum of unitary terms, it is not trivial that the result is also unitary. Thus, we provide in Appendix \ref{unitary_appendix_dec} the proof that the matrix is indeed unitary, using the property that the generalized Pauli matrices form an orthonormal basis. Therefore, proving that the proposed matrix is unitary allows us to easily implement it as a gate in a qudit quantum circuit.

\subsection{Correctness of the unitaries}
As stated in Section \ref{sec:Preliminaries}, the quantum state obtained after applying the encryption unitary to the initial state should be an encrypted state that reveals no information to an external observer. Furthermore, applying the decryption unitary to the encrypted state must recover the original information stored in the data qudit to the chosen party, specifically \(S_1\) in our case. If both requirements are satisfied, we assume that the two unitaries are suitable within the scope of the protocol. To prove this, we propose \Cref{th_Tr_Ni} and \Cref{trace_gen_bell}. We provide the corresponding proofs in Appendices \ref{proof_th_Tr_Ni} and \ref{proof_trace_gen_bell}.

\begin{theorem}
\label{th_Tr_Ni}
For any operator \(O_1,O_2 \in \mathcal{M}_{\mathbb{C}^{d}\times\mathbb{C}^{d}}\),
    \[\text{Tr}_B\Big((O_1\otimes I_B)\ket{\Phi_d}_{A,B}\bra{\Phi_d}_{A,B}(O_2^{\dagger }\otimes I_B)\Big) = \frac{1}{d}O_1O_2^{\dagger}.\]
\end{theorem}

\begin{theorem}
\label{trace_gen_bell}
For all \(k,l,m,n \in \{0,\dots,d-1\}\),
\[\text{Tr}\Big((X_d^{\scriptscriptstyle k}Z_d^{\scriptscriptstyle l}\otimes I_d)\ket{\Phi_d}\bra{\Phi_d}(Z_d^{\scriptscriptstyle -n}X_d^{\scriptscriptstyle -m}\otimes I_d)\Big) = \delta_{k,m}\delta_{l,n}.\]
\end{theorem}

To analyze whether the encrypted state provides no information about the data qudit, we compute the partial trace over a specific subsystem and show that it equals the maximally mixed state. In this scenario, we can be sure that no adversary can retrieve any information. We perform the computation on subsystem \(S_1\). We proved that the expected result is obtained without imposing any additional constraints on the other subsystems. Thus, the same argument can be applied to any other \(S_i\).

To show that the proposed decryption gate performs the correct operation, we begin with the encrypted state and analyze its evolution. Since we selected the subsytem \((S_1)\) as the designated party, we demonstrate that the final state consists of a generalized Bell state \(\ket{\Phi_d}\) shared between qudits \(\ket{\psi}_A, \text{ and } \ket{\psi}_{N1}\), while the data qudit is transferred to \(\ket{\psi}_{S_1}\). All other pairs are likewise left in the generalized Bell state \(\ket{\Phi_d}\). Since this result is obtained without imposing any additional constraints, we can assume, without loss of generality, that the proposed unitary correctly implements the desired operation.

\section{\label{sec:impl}Implementation analysis}
In this section we present how the unitaries described for encryption and decryption could be implemented in a quantum circuit.
\subsection{Implementation of encryption operator}
The circuit that implements the encryption unitary is constructed by two separate circuits. Each of the two is implementing one of \(V(P_X), V(P_Z)\), considering the expression of \(V(P)\) from \eqref{new_operator}. The circuit that implements \(V(P_Z)\) is presented in Figure \ref{FIG_impl_PZ}, while the one that implements \(V(P_X)\) in Figure \ref{FIG_impl_PX}.
\begin{figure}[t]
    
    \begin{center}
    
        \begin{quantikz}[column sep=6pt, row sep=0.1cm, scale=0.8, transform shape]
        \lstick{$A$}  &\control{} \vqw{1}& \qw & \qw &\qw &\qw &\qw & \control{} \vqw{1} & \qw&
        \\
        
        \lstick{${S_1}$}& \gate[style={inner sep=-1pt}]{X_d^{\scriptscriptstyle\phantom{\dagger}}}& \control{} \vqw{1} &\qw &\qw &\qw &\control{} \vqw{1}        & \gate[style={inner sep=-1pt}]{X_d^{\scriptscriptstyle\dagger}}  & \qw &\\
        
        \lstick{${S_2}$}& \qw&\gate[style={inner sep=-1pt}]{X_d^{\scriptscriptstyle\phantom{\dagger}}}&\qw & \qw & \qw & \gate[style={inner sep=-1pt}]{X_d^{\scriptscriptstyle\dagger}}&\qw &\qw&
        \\
        
        \dots\\

        \lstick{${S_{n-1}}$}& \qw & \qw & \control{} \vqw{1}& \qw  &\control{} \vqw{1}&\qw &\qw &\qw&
        \\
        
        \lstick{${S_ n}$}& \qw & \qw &\gate[style={inner sep=-1pt}]{X_d^{\scriptscriptstyle\phantom{\dagger}}}& \gate{Q} &\gate[style={inner sep=-1pt}]{X_d^{\scriptscriptstyle\dagger}} & \qw &\qw &\qw&
        \end{quantikz}
    \end{center}
    \caption{Quantum circuit for implementation of \(V(P)\) from \eqref{new_operator}, with \(P=P_Z\). The circuit implements the controlled gate \(C(X_d)\), as defined in \eqref{def_CXd}.}
    \label{FIG_impl_PZ}
\end{figure}
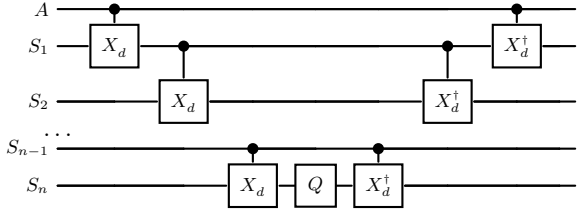
The circuit in Figure \ref{FIG_impl_PZ} consists of \(n \times\) \(C(X_d)\) gates, \(n \times\) \(C(X_d)^{\dagger}\) gates and the single qudit gate \(Q\). The single qudit gate \(Q\) is given as
\begin{equation}
    Q\ket{k} = \frac{1}{\sqrt{d}}\sum_{j=0}^{d-1}e^{-i\frac{\pi j(j+d\%2)}{d}}e^{\frac{2i\pi jk}{d}}\ket{k}.
\end{equation}
It is obvious that the gate is a diagonal matrix, which could be implemented using \(d\) single-qudit gates, as presented in~[\onlinecite{Nikolaeva2024}]. Thus, we express the \(Q\) matrix as \(\text{diag}(\phi_0, \phi_1,\dots,\phi_{d-1})\), with each angle of rotation being expressed as
\begin{equation}
    \phi_k = \arg\Big(\sum_{j=0}^{d-1}e^{-i\frac{\pi j(j+d\%2)}{d}}e^{\frac{2i\pi jk}{d}}\Big).
\end{equation}
The state of the quantum system before applying the gate \(Q\) is given as
\[
\sum_{i_A, p_1,\dots,p_n=0}^{d-1}\ket{i_A}\ket{i_A \oplus_d p_1}\dots\ket{i_A \oplus_d p_1 \oplus_d\dots \oplus_d p_n},
\]
where we have been interested just in the quantum states, because the amplitudes are not relevant for proving the effect of the circuit. The state of the circuit after the \(Q\) gate is 
\begin{align*}
\sum_{i_A, p_1,\dots,p_n=0}^{d-1}\sum_{j=0}^{d-1}e^{-i\frac{\pi j(j+d\%2)}{d}}\omega^{j(i_A+p_1+\dots+p_n)}\cdot\\
\cdot \ket{i_A}\dots\ket{i_A \oplus_dp_1\oplus_d\dots \oplus_d p_n}.
\end{align*}
After performing the ladder composed of \(C(X_d)^{\dagger}\), the state at the end of the circuit is 
\[
\sum_{j=0}^{d-1}e^{-i\frac{\pi j(j+d\%2)}{d}}\omega^{ji_A}\ket{i_A}\omega^{jp_1}\ket{p_1}\dots\omega^{jp_n}\ket{p_n},
\]
which is exactly the implementation of the \(V(P_Z)\) operator because \(Z_d^{k}\ket{j} = \omega^{jk}\ket{j}\). Thus, the implementation of the \(V(P_Z)\) operator can be implemented using \(2n\) two-qudit gates and \(d-1\) single qudit gates, as the gate corresponding to \(\phi_0\) is the identity. The implementation of \(V(P_X)\) follow the same principle, with the adition of \(n\times F\) and \(n\times F^{\dagger}\)gates, as we recall from \eqref{XZ_TRANS} that \(X_d = F^{\dagger}Z_dF\). Thus, the entire encryption unitary could be implemented using 
\begin{equation}
    \label{NE1Q}
    N_{E_{2Q}} = 4n 
\end{equation}
two-qudit gates and 
\begin{equation}
    N_{E_{1Q}} = 2n+2(d-1)
\end{equation}
single-qudit gates.
\begin{figure}[t]
    \begin{center}
        \begin{quantikz}[column sep=3.3pt, row sep=0.1cm, scale=0.8, transform shape]
        \lstick{$A$} &\gate{F} &\control{} \vqw{1}& \qw & \qw &\qw &\qw &\qw & \control{} \vqw{1} & \qw&\gate{F^{\scriptscriptstyle\dagger}}&&
        \\
        
        \lstick{${S_1}$}&\gate{F} &\gate[style={inner sep=-1pt}]{X_d^{\scriptscriptstyle\phantom{\dagger}}}& \control{} \vqw{1} &\qw &\qw &\qw &\control{} \vqw{1}        & \gate[style={inner sep=-1pt}]{X_d^{\scriptscriptstyle\dagger}}  & \qw &\gate{F^{\scriptscriptstyle\dagger}}&&
        \\
        
        \lstick{${S_2}$}&\gate{F} &\qw&\gate[style={inner sep=-1pt}]{X_d^{\scriptscriptstyle\phantom{\dagger}}}&\qw & \qw & \qw & \gate[style={inner sep=-1pt}]{X_d^{\scriptscriptstyle\dagger}}&\qw &\qw&\gate{F^{\scriptscriptstyle\dagger}}&&
        \\
        
        \dots\\

        \lstick{${S_{n-1}}$}&\gate{F} &\qw & \qw & \control{} \vqw{1}& \qw  &\control{} \vqw{1}&\qw &\qw &\qw&\gate{F^{\scriptscriptstyle\dagger}}&&
        \\
        
        \lstick{${S_ n}$}& \gate{F}&\qw & \qw &\gate[style={inner sep=-1pt}]{X_d^{\scriptscriptstyle\phantom{\dagger}}}& \gate{Q} &\gate[style={inner sep=-1pt}]{X_d^{\scriptscriptstyle\dagger}} & \qw &\qw &\qw&\gate{F^{\scriptscriptstyle\dagger}} &&
        \end{quantikz}
    \end{center}
    \caption{Quantum circuit for implementation of \(V(P)\) from \eqref{new_operator}, with \(P=P_X\). The circuit implements the controlled gate \(C(X_d)\), as defined in \eqref{def_CXd}}
    \label{FIG_impl_PX}
\end{figure}
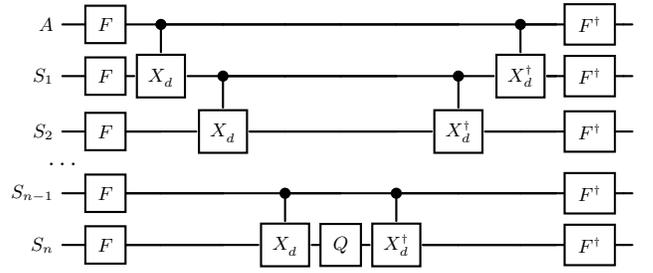

\subsection{Implementation of decryption operator}
As the decryption operator has been constructed as the one in~[\onlinecite{y4y1-1ll6}], the implementation of the unitary follows the same principle. Thus, we write the decryption unitary as 
\begin{equation}
    U_{dec_d}^{(n)} = c_{00}\cdot {\overline{T}^{\dagger}_{S_1,N_1}}\cdot T_{d^2-1}\cdot ... \cdot T_2\cdot  T_1\cdot\overline{T}_{S_1,N_1}.
\end{equation}
The gate \(\overline{T}_{S_1,N_1}\) performs the mapping
\begin{equation}
\overline{T}_{S_1,N_1}(X_d^{k}Z_d^{l}\otimes I_d)\ket{\Phi_d} = \ket{k}\ket{l},
\end{equation}
and can be easily implemented using a \(C(X_d)^{\dagger}\) gate, a \(F^{\dagger}\) gate, and two SWAP gates. Each of the \(d^2\) T gates has the following construction
\begin{align}
\label{Tk}
    T_{\overline{kl}} &= \frac{c_{kl}}{c_{00}}\ket{k}\ket{l}\bra{k}\bra{l}_{S_1,N_1} \otimes \Big(\bigotimes_{j=2}^{n}{X_d^{\scriptscriptstyle(N_j)}}^{k}{Z_d^{\scriptscriptstyle(N_j)}}^{-l}\Big) + \notag\\
    &+ \Big(I_d^{\scriptscriptstyle(S_1,N_1)} - \ket{k}\ket{l}\bra{k}\bra{l}_{S_1,N_1} \otimes I_d^{\scriptscriptstyle(N_2,\dots,N_n)}\Big),
\end{align}
where we have noted by \(\overline{kl} = kd+l\).
We can observe that when \(k=l=0\), we have \(T_{\overline{00}} = I_d\). Thus, we need to implement just \(d^2-1\) gates, each being constructed as the circuit from Figure \ref{FIG_impl_T}. Each of the \(X_d\) and \(Z_d\) gates is $p$-controlled by the qudits \(S_1\) and \(N_1\). The control level for the qudit \(S_1\) is given by the value of \(k\), while for the qudit \(N_1\) it is given by the value of \(l\). To mimic the control level for each of the \(T_{\overline{kl}}\) gates, we can use \((2k)\times X_d\) gates for qudit \(S_1\) and \((2l)\times\) \(X_d\) gates for qudit \(N_1\). 
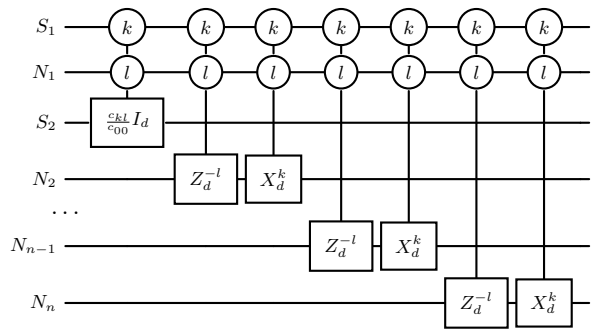
\begin{figure}[t]
    \begin{center}
        \begin{quantikz}[column sep=3.3pt, row sep=0.1cm, scale=0.8, transform shape]

        \lstick{$S_1$} &\gate[style={circle, inner sep=-2pt, text=white}]{k} \vqw{1}&\gate[style={circle, inner sep=-2pt, text=white}]{k} \vqw{1}&\gate[style={circle, inner sep=-2pt, text=white}]{k} \vqw{1}&\gate[style={circle, inner sep=-2pt, text=white}]{k} \vqw{1}&\gate[style={circle, inner sep=-2pt, text=white}]{k} \vqw{1}&\gate[style={circle, inner sep=-2pt, text=white}]{k} \vqw{1}&\gate[style={circle, inner sep=-2pt, text=white}]{k} \vqw{1}&&
        \\
        
        \lstick{$N_1$} &\gate[style={circle, inner sep=-2pt, text=white}]{l} \vqw{1}&\gate[style={circle, inner sep=-2pt, text=white}]{l} \vqw{2}&\gate[style={circle, inner sep=-2pt, text=white}]{l} \vqw{2}&\gate[style={circle, inner sep=-2pt, text=white}]{l} \vqw{4}&\gate[style={circle, inner sep=-2pt, text=white}]{l} \vqw{4}&\gate[style={circle, inner sep=-2pt, text=white}]{l} \vqw{5}&\gate[style={circle, inner sep=-2pt, text=white}]{l} \vqw{5}&&
        \\
        
        \lstick{$S_2$}&\gate[style={inner xsep=-1pt}]{{\scriptscriptstyle \frac{c_{kl}}{c_{00}}}I_d}&&&&&&&&
        \\
        
        \lstick{$N_2$}&&\gate{Z_d^{-l}}&\gate{X_d^{k}}&&&&&&
        \\
        
        \dots\\
        \lstick{$N_{n-1}$}&&&&\gate{Z_d^{-l}}&\gate{X_d^{k}}&&&&
        \\
        
        \lstick{$N_n$}&&&&&&\gate{Z_d^{-l}}&\gate{X_d^{k}}&&

        \end{quantikz}
    \end{center}
    \caption{Quantum circuit for implementation of \(T_{\overline{kl}}\) from \eqref{Tk}.The control values for the \(C_k(X_d)\) and \(C_l(Z_d)\) are marked by the values in the circles, using the definitions of the quantum gates from \eqref{k_controlled}.}
    \label{FIG_impl_T}
\end{figure}
Thus, all of the \(T_{\overline{kl}}\) gates can be implemented using \((2n-1)d^2(d-1)\) single qudit gates. Also, considering that we apply an individual double-controlled gate for each of the \(X_d^k\) and \(Z_d^{-l}\), the required number of double-controlled qudit gates is equal to \((2n-1)(d^2-1)\). 

To count the total number of single qudit gates that are needed to perform the decryption unitary, we sum up the number of gates for all of the \(T_{\overline{kl}}\) with \(2\) gates needed for the \(\overline{T}, \overline{T}^{\dagger}\) and with the two aditional \(F\) gates needed to construct the \(C\) matrix. Thus, we obtain a total of  
\begin{equation}
    N_{D_{1Q}} = 2+(2n-1)d^2(d-1)
\end{equation}
single-qudit gates. 
The total number of two-qudit gates required is obtained by summing the \(9\) gates needed to implement \(\overline{T}, \overline{T}^{\dagger}\), i.e. \(4\times C(X_d) + 4\times\text{SWAP} + 1\times\text{SWAP}\) (the additional SWAP and $2\times C(X_d)$ comes from the additional SWAP and $C$ gate, that are presented in the construction of the decryption operator \eqref{U_dec_mine}) and the \(\alpha \times (2n-1)(d^2-1)\) gates for implementing all of the \(T_{\overline{kl}}\). We can get an estimate on the parameter \(\alpha\) from~[\onlinecite{PhysRevA.87.012325}], as any double-controlled unitary can be implemented with at most \(8(d-1)\) two-qudit gates. Thus, the total number of two-qudit gates required for implementing the decryption operator is at most
\begin{equation}
    N_{D_{2Q}} = 9+8(2n-1)(d^3-d^2-d+1).
\end{equation}

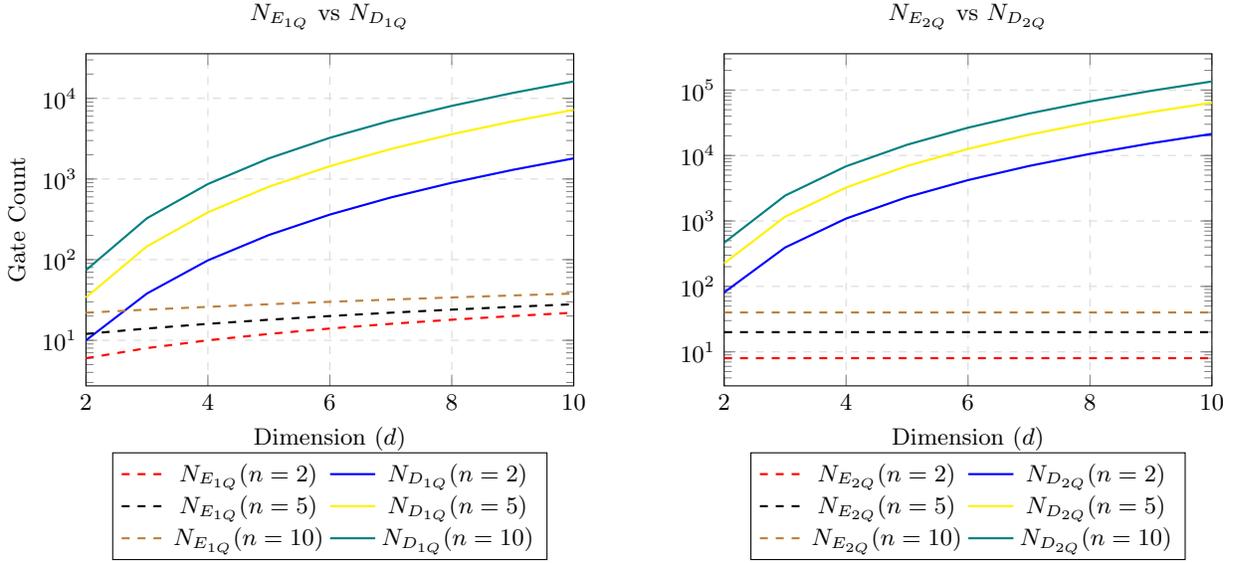
\begin{figure*}[ht]
\centering
\begin{tikzpicture}
    \begin{groupplot}[
        group style={
            group size=2 by 1, 
            horizontal sep=2cm,
        },
        width=0.45\linewidth, 
        height=6cm,
        ymode=log,            
        grid=major,
        grid style={dashed, gray!30},
        xlabel={Dimension ($d$)},
        xmin=2, xmax=10,
        legend style={font=\small, at={(0.5,-0.2)}, anchor=north, legend columns=2},
        cycle list name=color list 
    ]

    \nextgroupplot[title={$N_{E_{1Q}}$ vs $N_{D_{1Q}}$}, ylabel={Gate Count}]
        \addplot+ [dashed, thick] table [x=d, y=NE1Q_n2, col sep=comma] {imgs/gates_1q.csv};
        \addlegendentry{$N_{E_{1Q}} (n=2)$}
        \addplot+ [solid, thick, mark=none] table [x=d, y=ND1Q_n2, col sep=comma] {imgs/gates_1q.csv};
        \addlegendentry{$N_{D_{1Q}} (n=2)$}
        
        \addplot+ [dashed, thick] table [x=d, y=NE1Q_n5, col sep=comma] {imgs/gates_1q.csv};
        \addlegendentry{$N_{E_{1Q}} (n=5)$}
        \addplot+ [solid, thick, mark=none] table [x=d, y=ND1Q_n5, col sep=comma] {imgs/gates_1q.csv};
        \addlegendentry{$N_{D_{1Q}} (n=5)$}

        \addplot+ [dashed, thick] table [x=d, y=NE1Q_n10, col sep=comma] {imgs/gates_1q.csv};
        \addlegendentry{$N_{E_{1Q}} (n=10)$}
        \addplot+ [solid, thick, mark=none] table [x=d, y=ND1Q_n10, col sep=comma] {imgs/gates_1q.csv};
        \addlegendentry{$N_{D_{1Q}} (n=10)$}

    \nextgroupplot[title={$N_{E_{2Q}}$ vs $N_{D_{2Q}}$}]
        \addplot+ [dashed, thick] table [x=d, y=NE2Q_n2, col sep=comma] {imgs/gates_2q.csv};
        \addlegendentry{$N_{E_{2Q}} (n=2)$}
        \addplot+ [solid, thick, mark=none] table [x=d, y=ND2Q_n2, col sep=comma] {imgs/gates_2q.csv};
        \addlegendentry{$N_{D_{2Q}} (n=2)$}
        
        \addplot+ [dashed, thick] table [x=d, y=NE2Q_n5, col sep=comma] {imgs/gates_2q.csv};
        \addlegendentry{$N_{E_{2Q}} (n=5)$}
        \addplot+ [solid, thick, mark=none] table [x=d, y=ND2Q_n5, col sep=comma] {imgs/gates_2q.csv};
        \addlegendentry{$N_{D_{2Q}} (n=5)$}

        \addplot+ [dashed, thick] table [x=d, y=NE2Q_n10, col sep=comma] {imgs/gates_2q.csv};
        \addlegendentry{$N_{E_{2Q}} (n=10)$}
        \addplot+ [solid, thick, mark=none] table [x=d, y=ND2Q_n10, col sep=comma] {imgs/gates_2q.csv};
        \addlegendentry{$N_{D_{2Q}} (n=10)$}
        
    \end{groupplot}
\end{tikzpicture}
\caption{Comparison of $N_E$ and $N_D$ gate counts for varying dimensions $d$ and qudit numbers $n$. Left: Single-qudit gate scaling. Right: Two-qudit gate scaling.}
\label{fig:gate_comparison}
\end{figure*}
Comparing the scaling of the required resources, we observe that for the encryption gate, our results are similar to those in~[\onlinecite{y4y1-1ll6}]. The number of two-qudit gates needed to perform the encryption does not depend on the dimension of the quantum system, as presented in \eqref{NE1Q}. We observe a linear scaling with dimension for the necessary single-qudit quantum gates. The decryption operation introduces a gate complexity that scales \(\mathcal{O}(nd^3)\) for both single and two-qudit gates, as it can be observed in \Cref{fig:gate_comparison}. Applying the protocol for qudits does not impose an increase in complexity compared to the case of qubits, because in the original reference, the number of gates that are equivalent to \(T_{\overline{kl}}\) that were applied is equal to \(3\). The difference arises from the decomposition of double-controlled qudit gates, whose decomposition into two-qudit gates could be implemented with \(\mathcal{O}(d)\) gate complexity. The difference in the number of single-qudit gates required to perform the decryption matrix arises from the fact that for the case of \(d=2\), the $p$-controlled gate is equivalent to the controlled gate.

\section{\label{sec:concl}Conclusions}
We have solved the open problem posed by Yamaguchi and Kempf and proved that the protocol generalizes successfully to qudits. As the Pauli operators are not hermitic for any \(d\geq3\),  we have constructed an operator based on CAZAC sequences. The proposed encryption operator is a consistent generalization of the encryption gate in the original paper, as we have proved. The decryption matrix follows the construction architecture from the qubit case and is adjusted to account for the differences introduced by qudits. The additional overhead introduced by the quantum system's dimensionality is linear in the qudit dimension for both encryption and decryption. 

\begin{acknowledgments}
We thank D. Maimu\c{t}, A. Frunz\u{a}, and G. Te\c{s}eleanu for their valuable comments, and R. Ionicioiu for insightful discussions and technical suggestions.
\end{acknowledgments}

\section*{Data Availability}
The code used for the construction of the gates, the simulation of the protocol, and the testing of the proved Lemmas and Theorems is available at the following repository \url{https://github.com/FipNad/Cloning-encrypted-quantum-states-in-arbitrary-dimensions}.

\appendix

\section{Useful Lemmas for multi-dimensional quantum states}
\begin{lema}
\label{th_ricochet}
    For any operator \(U\in \mathcal{M}_{\mathbb{C}^{d}\times \mathbb{C}^d}\), the following equality is true.
    \[
    (U\otimes I_d)\ket{\Phi_d} = (I_d\otimes U^{\mathsf{T}})\ket{\Phi_d}
    \]
\end{lema}
\begin{proof}
We prove the lemma by computing the LHS and the RHS and observing that they are equal. Firstly, we express \(U\) by means of an orthonormal basis. Thus, we can express \(U\) and its transpose as:
\[
    U = \sum_{i,j=0}^{d-1}U_{ij}\ket{i}\bra{j}, \text{ and } U^{\mathsf{T}} = \sum_{i,j=0}^{d-1}U_{ji}\ket{i}\bra{j}.
\]
The LHS has the following derivation
\begin{align*}
(U\otimes I_d)\ket{\Phi_d} &= \frac{1}{\sqrt{d}}\sum_{p=0}^{d-1}\Big(\sum_{i,j=0}^{d-1}U_{ij}\ket{i}\braket{j}{p}\Big)\ket{p} \\
&= \frac{1}{\sqrt{d}}\sum_{p,i=0}^{d-1}U_{ip}\ket{i}\ket{p}.
\end{align*}
Similarly, the RHS of the equation has the following value
\begin{align*}
(I_d\otimes U^{\mathsf{T}})\ket{\Phi_d} &= \frac{1}{\sqrt{d}}\sum_{p=0}^{d-1}\ket{p}\Big(\sum_{i,j=0}^{d-1}U_{ji}\ket{i}\braket{j}{p}\Big) \\
&= \frac{1}{\sqrt{d}}\sum_{p,i=0}^{d-1}U_{pi}\ket{p}\ket{i}.
\end{align*}
We can clearly observe that the two states are equal; thus, we can claim that the statement is true.
\end{proof}

\begin{lema}
\label{th_swap}
Let \(\ket{\psi}\) be a qudit state, \(\ket{\Phi_d}\) be the generalized Bell state, \(X_d, Z_d\) the generalized Pauli operators, and $C = \left(I_d \otimes I_d\otimes F^2\right) \cdot \left(I_d\otimes\sum_{c=0}^{d-1}X_d^2 \otimes \ket{c}\bra{c}\right)$. 
    For any \(d\geq2\), the following equality is true.
    \begin{align*}
    C \cdot \frac{1}{d}\sum_{m,n=0}^{d-1}\Big(X_d^{m}Z_d^{n}\Big)\ket{\psi} \otimes \Big(X_d^{m}Z_d^{n} &\otimes I_d\Big)\ket{\Phi_d} = \\
    =\frac{1}{\sqrt{d}}\Big(\sum_{p=0}^{d-1}\ket{p}\ket{p}\Big)\ket{\psi}
    \end{align*}
\end{lema}
\begin{proof}
We start by using the Lemma \ref{th_ricochet} for rewriting the left-hand side of the equation, using the property that \(Z_d^{\mathsf{T}} = Z_d\), and \(X_d^{\mathsf{T}} = X_d^{-1}\).
\begin{align*}
\text{LHS} &= C\frac{1}{d}\sum_{m,n=0}^{d-1}\Big(X_d^{m}Z_d^{n}\otimes (X_d^{m}Z_d^{n}\otimes I_d)\Big)\ket{\psi} \otimes\ket{\Phi_d} \\
&= C\frac{1}{d}\sum_{m,n=0}^{d-1} \Big(X_d^{m}Z_d^{n}\otimes (I_d \otimes Z_d^{n}X_d^{-m})\Big)\ket{\psi}\otimes\ket{\Phi_d}\\
&\text{We express $\ket{\psi}_A$ using an orthonormal basis}\\
&\text{$\ket{\psi}_{A} = \sum_i Q_i\ket{i}$, and thus obtain}\\
&= C\frac{1}{\sqrt{d^3}}\sum_{m,n,i,p=0}^{d-1}Q_iX_d^{m}Z_d^{n}\ket{i} \otimes\ket{p} \otimes Z_d^{n}X_d^{-m}\ket{p}\\
&= C\frac{1}{\sqrt{d^3}}\sum_{m,n,i,p=0}^{d-1}Q_i\omega^{n(i+p-m)}\ket{i\oplus_dm} \ket{p} \ket{p\ominus_dm}\\
&\text{We use the property: $\sum_{n}\omega^{n(i+p-m)}=d\cdot\delta_{{i+p-m},0}$}\\
&= C\frac{1}{\sqrt{d}}\sum_{i,p=0}^{d-1}Q_i\ket{i\oplus_d\oplus_di\oplus_dp}\ket{p}\ket{d \ominus_d i}\\
&=\frac{1}{\sqrt{d}}\sum_{i,p,c=0}^{d-1}Q_i\ket{2i\oplus_d p}X_d^{2c}\ket{p}\ketbra{c}F^2\ket{d\ominus_di}\\
&= \frac{1}{\sqrt{d}}\sum_{i,p,c=0}^{d-1}Q_i\ket{2i\oplus_dp}X_d^{2c}\ket{p}\ketbra{c}\ket{i}\\
&=\frac{1}{\sqrt{d}}\sum_{i,p=0}^{d-1}Q_i\ket{2i\oplus_dp}\ket{2i\oplus_dp}\ket{i}
=\ket{\Phi_d}\ket{\psi}.
\end{align*}
We showed that the LHS of the equation is equivalent to the RHS; therefore, we can affirm that the lemma is proved.
\end{proof}

\begin{lema}
\label{th_orthonormal_basis}
Let \(\ket{\Phi_d}\) be the generalized Bell state, and \(X_d, Z_d\) the generalized Pauli operators. 
For any 
    \(k_1,k_2, k_3, k_4 \in \{0,\dots,d-1\}\), the following identity holds.
    \[\bra{\Phi_d}(Z_d^{\scriptscriptstyle-k_2}X_d^{\scriptscriptstyle-k_1}\otimes I_d)\cdot (X_d^{k_3}Z_d^{k_4}\otimes I_d)\ket{\Phi_d} = \delta_{k_1,k_3}\delta_{k_2,k_4}\]

\end{lema}
\begin{proof}
We prove the Lemma by calculating the inner product 
\[\bra{\Phi_d}(Z_d^{-k_2}X_d^{-k_1}\otimes I_d) (X_d^{k_3}Z_d^{k_4}\otimes I_d)\ket{\Phi_d}.\]
\begin{align*}
    \beta &= \bra{\Phi_d}(Z_d^{-k_2}X_d^{-k_1}\otimes I_d)\cdot (X_d^{k_3}Z_d^{k_4}\otimes I_d)\ket{\Phi_d}\\
    &=\frac{1}{d}\sum_{i,j=0}^{d-1}\bra{ii}Z_d^{\scriptscriptstyle-k_2}X_d^{\scriptscriptstyle-k_1}\otimes I_d) \cdot (X_d^{\scriptscriptstyle k_3}Z_d^{\scriptscriptstyle k_4}\otimes I_d)\ket{jj}\\
    &=\frac{1}{d}\sum_{i,j=0}^{d-1}\omega^{-ik_2}\bra{i\ominus_dk_1}\bra{i}\cdot\omega^{k_4j}\ket{j\oplus_dk_3}\ket{j}\\
    &=\frac{1}{d}\sum_{i,j=0}^{d-1}\omega^{-ik_2+k_4j}\braket{i\ominus_dk_1}{j\oplus_dk_3}\braket{i}{j}\\
    &=\frac{1}{d}\sum_{i=0}^{d-1}\omega^{-i(k_2-k_4)}\braket{i\ominus_dk_1}{i\oplus_dk_3}\\
    &=\frac{1}{d}\Big(\sum_{i=0}^{d-1}1\Big)\delta_{k_1,k_3}\delta_{k_2,k_4} = \delta_{k_1,k_3}\delta_{k_2,k_4}\\
\end{align*}
\end{proof}

\begin{lema}
\label{th_k_k1_=k}
Let \(\ket{\Phi_d}\) the generalized Bell state, and \(X_d, Z_d\) the generalized Pauli operators. 
For any
\(k_1,k_2, k_3, k_4 \in \{0,\dots,d-1\}\), the following identity is true
\[
\Pi_{k_1,k_2} \cdot \Pi_{k_3,k_4} = \delta_{k_1,k_3}\delta_{k_2,k_4}\Pi_{k_1,k_2},
\]
given the notation of the puter product \(\Pi_{i,j} = (X_d^{i}Z_d^{j}\otimes I_d)\ket{\Phi_d}\bra{\Phi_d}(Z_d^{-j}X_d^{-i }\otimes I_d)\).
\end{lema}
\begin{proof}
We prove this Lemma by computing the LHS and observing it is equal to the RHS.
\begin{align*}
LHS &= \Big( X_d^{k_1} Z_d^{k_2} \otimes I_d \Big) \ket{\Phi_d}\bra{\Phi_d} \Big( Z_d^{-k_2} X_d^{-k_1} \otimes I_d \Big) 
\cdot\\
 &\Big( X_d^{k_3} Z_d^{k_4} \otimes I_d \Big) \ket{\Phi_d}\bra{\Phi_d} \Big( Z_d^{-k_4} X_d^{-k_3} \otimes I_d \Big)\\
 &\text{Using the result from Lemma \ref{th_orthonormal_basis}},\\
 &=\Big( X_d^{k_1} Z_d^{k_2} \otimes I_d \Big) \ket{\Phi_d}\bra{\Phi_d}\Big( Z_d^{-k_2} X_d^{-k_1} \otimes I_d \Big)\\
 &=\Pi_{k_1,k_2}
\end{align*}
\end{proof}

\begin{lema}
\label{sum_basis_Psi_k1k2}
Let \(\ket{\Phi_d}\) be the generalized Bell state, and \(X_d, Z_d\) the generalized Pauli operators. 
For any value of \(d\geq2\), the following equality holds.
\[
\sum_{m,n=0}^{d-1}(X_d^{m}Z_d^{n}\otimes I_d)\ket{\Phi_d}\bra{\Phi_d}(Z_d^{-n}X_d^{-m}\otimes I_d) = I_d^{\otimes(2)}.
\] 
\end{lema}
\begin{proof}
We prove this Lemma by computing the LHS and observing it is equal to the RHS.
\begin{align*}
    LHS &= \sum_{m,n=0}^{d-1}(X_d^{m}Z_d^{n}\otimes I_d)\ket{\Phi_d}\bra{\Phi_d}(Z_d^{-n}X_d^{-m}\otimes I_d)\\
    &=\sum_{k_1,k_2=0}^{d-1}\Pi_{k_1,k2}\\
    &\text{Since in Lemma \ref{th_orthonormal_basis} we proved that}\\
    &\text{$X_d^{m}Z_d^{n}\ket{\Phi_d}$ is an orthonormal basis}\\
    &=I_d^{\otimes (2)}
\end{align*}
\end{proof}

\begin{lema}
\label{th_unitary_P}
For any value of
\(m\in \{0,\dots,d-1\}\), the following equality is true
\[\sum_{j=0}^{d-1}e^{-i\frac{\pi (j+m)((j+m)+d\%2)}{d}} e^{i\frac{\pi j(j+d\%2)}{d}} = d\cdot\delta_{m,0},\]
where \(d\geq2\) is fixed as the qudit dimension.
\end{lema}
\begin{proof}
We prove this Lemma by computing the LHS and observing it is equal to the RHS.
The exponential could be simplified as 
\begin{align*}
    E &= e^{-i\frac{\pi (j+m)((j+m)+d\%2)}{d}} e^{i\frac{\pi j(j+d\%2)}{d}}\\
    &=e ^{-i\frac{\pi}{d}(m^2+2jm +m\cdot (d\%2))}\\
    &=e^{-i\frac{\pi}{d}(m^2+m\cdot(d\%2))}e^{-i\frac{2\pi jm}{d}}.
\end{align*}
Thus, the initial sum can be written as 
\begin{align*}
\sum_{j=0}^{d-1}E &=e^{-i\frac{\pi}{d}(m^2+m\cdot(d\%2))}\sum_{j=0}^{d-1}e^{-i\frac{2\pi jm}{d}}\\
&\text{Because $\sum_{j}e^{-i\frac{2\pi jm}{d}} = d\cdot\delta_{m,0}$}\\
&=d\delta_{m,0}
\end{align*}
    
\end{proof}

\section{Unitarity of encryption operator}
\label{unitary_appendix_enc}
To prove that \(U_{enc_d}^{(n)}\) is unitary, we show that each of \(A_X\) and \(A_Z\) is unitary.
We prove \(A_X\) is unitary, by performing the notation \(\alpha_k = e^{-i\frac{\pi ik(k+d\%2)}{d}}\):
\begin{align*}
    A_XA_X^{\dagger} &= \frac{1}{d}\sum_{k,j=0}^{d-1}\alpha_{k}P_X^k \alpha_{j}^{*}P_X^{-j}\\
    &=\frac{1}{d}\sum_{k,j=0}^{d-1}\alpha_{k}\alpha_{j}^{*}P_X^{(k-j)}\\
    &\text{We perform the change of variable $k-j \mapsto m$}\\
    &=\frac{1}{d}\sum_{m=0}^{d-1}\Big(\sum_{j=0}^{d-1}\alpha_{j+m}\alpha_{j}^{*}\Big)P_X^{m}\\
    &\text{Using Lemma \ref{th_unitary_P}, we simplify}\\
    &=\frac{1}{d}\sum_{m=0}^{d-1}\delta_{m,0} P_X^m\\
    &=P_X^0 = I_d.
\end{align*}

Following the same principle, we show that \(A_Z\) is unitary, using the same notation \(\alpha_k = e^{-i\frac{\pi ik(k+d\%2)}{d}}\):
\begin{align*}
    A_ZA_Z^{\dagger} &= \frac{1}{d}\sum_{k,j=0}^{d-1}\alpha_{k}P_Z^k \alpha_{j}^{*}P_Z^{-j}\\
    &=\frac{1}{d}\sum_{k,j=0}^{d-1}\alpha_{k}\alpha_{j}^{*}P_Z^{(k-j)}\\
    &\text{We perform the change of variable $k-j \mapsto m$}\\
    &=\frac{1}{d}\sum_{m=0}^{d-1}\Big(\sum_{j=0}^{d-1}\alpha_{j+m}\alpha_{j}^{*}\Big)P_Z^{m}\\
    &\text{Using Lemma \ref{th_unitary_P}, we simplify}\\
    &=\frac{1}{d}\sum_{m=0}^{d-1}\delta_{m,0} P_Z^m\\
    &=P_Z^0 = I_d.
\end{align*}
We showed that \(A_X\), and \(A_Z\) are unitary matrices, thus their product, \(U_{enc_d}^{(n)}\)is a unitary matrix.

\section{Unitarity of decryption operator}
\label{unitary_appendix_dec}
If we recall \eqref{U_dec_mine}, the decryption operator can be written as 
\[U_{dec_d}^{(n)} = (B\cdot C)\otimes I_d^{\scriptscriptstyle \otimes (n-1)}\sum_{k,l=0}^{d-1}c_{kl}^{-1}O_{kl}\ket{\Phi_{d}}\bra{\Phi_{d}}O_{kl}^{\dagger}\otimes U_{kl},\]
where \(U_{kl} = \bigotimes_{j=2}^{n}X_d^{k}Z_d^{-l}\), \(O_{kl} = X_d^{k}Z_d^{l}\otimes I_d\), \(B = \text{SWAP}\), and \(C=(I_d \otimes I_d\otimes F^2\ \cdot (I_d\otimes\sum_{c=0}^{d-1}X_d^2 \otimes \ket{c}\bra{c})\). Since \((\text{SWAP}\cdot C)\otimes I_d^{\otimes (n-1)}\) is a unitary operator, we focus just on proving the other component, namely the operator \(A\) given as
\[A = \sum_{k,l=0}^{d-1}c_{kl}^{-1}O_{kl}\ket{\Phi_{d}}\bra{\Phi_{d}}O_{kl}^{\dagger}\otimes U_{kl}.\]
Because the operator is constructed by the summation of matrices that are obtained using tensor products of square matrices, it is obvious that the decryption operator is a square matrix. Thus, it is sufficient to show that \(AA^{\dagger} = I\).
\begin{align*}
    AA^{\dagger} &= \sum_{k,l=0}^{d-1}c_{kl}^{-1}O_{kl}\ket{\Phi_{d}}\bra{\Phi_{d}}O_{kl}^{\dagger}\otimes U_{kl}\cdot\\
    &\phantom{====}\cdot\sum_{m,n=0}^{d-1}{(c_{mn}^{-1})}^{*}O_{mn}\ket{\Phi_{d}}\bra{\Phi_{d}}O_{mn}^{\dagger}\otimes U_{mn}^{\dagger} \\
    &=\sum_{k,l,m,n=0}^{d-1}c_{kl}^{-1}{(c_{mn}^{-1})}^{*}O_{kl}\ket{\Phi_{d}}\bra{\Phi_{d}}O_{kl}^{\dagger}\cdot\\
    &\phantom{====}\cdot O_{mn}\ket{\Phi_{d}}\bra{\Phi_{d}}O_{mn}^{\dagger}\otimes U_{kl}U_{kmn}^{\dagger}\\
    &\text{Using the Lemma \ref{th_k_k1_=k},we simplify,}\\
    &=\sum_{k,l=0}^{d-1}c_{kl}^{-1}{(c_{kl}^{-1})}^{*}O_{kl}\ket{\Phi_{d}}\bra{\Phi_{d}}O_{kl}^{\dagger}\otimes U_{kl}U_{kl}^{\dagger}\\
    &\text{Since each $U_{kl} = X_d^{k}Z_d^{l}$ is unitary and $|c_{kl}|=1$,}\\
    &=\sum_{k,l=0}^{d-1}\Big(O_{kl}\ket{\Phi_{d}}\bra{\Phi_{d}}O_{kl}^{\dagger}\Big) \otimes I_d^{\otimes (n-1)}\\
    &\text{Using Lemma \ref{sum_basis_Psi_k1k2}, we find that}\\
    &=I_d^{\otimes (2)}\otimes _d^{\otimes (n-1)}
\end{align*}
We have shown that our proposed decryption operator is indeed unitary.

\section{Correctness proof}
\label{correctness}
The initial state of the quantum circuit is given as 
\begin{equation}
\ket{\psi}_{init} = \ket{\psi}_A\otimes\Big(\bigotimes_{i=1}^n\ket{\Phi_d}_{S_i,N_i}\Big).
\end{equation}
The encrypted state is obtained by applying the encryption unitary to the initial state. Thus, we obtain the state
\begin{align}  
    \ket{\psi}_{enc} &= U_{enc_{d}}^{(n)} \otimes I_d^{\otimes(n)}\ket{\psi}_{init} \notag\\
    &=\frac{1}{d}\sum_{k,l=0}^{d-1}c_{kl}\Big({X_d^{\scriptscriptstyle(A)}}^{k}{Z_d^{\scriptscriptstyle(A)}}^{l} \Big)\ket{\psi}_A \otimes \\
    &\phantom{===}\otimes\bigotimes_{j=1}^{n} \Big({X_d^{\scriptscriptstyle(S_i)}}^{k}{Z_d^{\scriptscriptstyle(S_i)}}^{l} \otimes I_d^{\scriptscriptstyle(N_i)}\Big)\ket{\Phi_d}\notag.
\end{align}

The encryption operator should produce a quantum state that does not leak any information regarding the original \(\ket{\psi}_A\) to any external party that might have access to any of the \(\ket{\psi}_{S_i}\). To show that the proposed unitary achieves this goal, we compute the partial trace on the \(S_1\) subsystem, and show that the result is the maximally mixed state \(\frac{1}{d}I_d\). We are tracing out the density matrix \(\rho_{enc} = \ket{\psi}_{enc}\bra{\psi}_{enc}\). We can express the density matrix of the state as
\begin{align}
    \rho_{enc} &= \frac{1}{d^2}\sum_{k,l,m,n=0}^{d-1}c_{kl}c_{mn}^{*}A_{kl}\ket{\psi}\bra{\psi}A_{mn}^{\dagger} \otimes \notag\\
    &\phantom{====}\otimes\bigotimes_{j=1}^{n}(S_{kl}\otimes I_d) \ket{\Phi_d}\bra{\Phi_d}(S_{mn}^{\dagger}\otimes I_d),
\end{align}
using the notations \(A_{kl} = {X_d^{\scriptscriptstyle(A)}}^{k}{Z_d^{\scriptscriptstyle(A)}}^{l}\), and \(S_{kl} = {X_d^{\scriptscriptstyle(S_i)}}^{k}{Z_d^{\scriptscriptstyle(S_i)}}^{l}\). The partial trace can be calculated as
\begin{align}
    \rho_{S_1} &= \text{Tr}_{\overline{S_1}}\rho_{enc},
\end{align}
where \(\overline{S_1} = A,S_2,\dots,S_n,N_1,\dots,N_n\). One intermediary step is to take the partial trace over all of \(S_i, N_i, i\in\{2,\dots,n\}\). Thus, the result is presented as follows.
\begin{align}
\label{rho_interm}
    \rho_{A,S_1,N_1} &= \text{Tr}_{S_2,\dots,S_n,N_2,\dots,N_n}(\rho_{enc}) \notag\\
    &=\frac{1}{d^2}\cdot\sum_{k,l,m,n=0}^{d-1}c_{kl}c_{mn}^{*}A_{kl}\ket{\psi}_A\bra{\psi}_AA_{mn}^{\dagger}\notag\otimes\\ 
    &\phantom{=}(S_{kl}\otimes I_d) \ket{\Phi_d}_{\scriptscriptstyle S_1,N_1}\bra{\Phi_d}_{\scriptscriptstyle S_1,N_1}(S_{mn}^{\dagger}\otimes I_d) \otimes\notag\\
    &\text{Tr}\Big(\bigotimes_{i=2}^{n}(S_{kl}\otimes I_d)\ket{\Phi_d}_{\scriptscriptstyle S_i,N_i}\bra{\Phi_d}_{\scriptscriptstyle S_i, N_i}(S_{mn}^{\dagger}\otimes I_d)\Big)\notag\\
    &\phantom{=}\text{Using the linearity of the trace and }\notag\\
    &\phantom{=}\text{the result from \Cref{trace_gen_bell},}\notag\\
    &=\frac{1}{d^2}\cdot\sum_{k,l=0}^{d-1}|c_{kl}|^2A_{kl}\ket{\psi}_A\bra{\psi}_AA_{kl}^{\dagger}\otimes\notag\\
    &\phantom{=}(S_{kl}\otimes I_d) \ket{\Phi_d}_{\scriptscriptstyle S_1,N_1}\bra{\Phi_d}_{\scriptscriptstyle S_1,N_1}(S_{kl}^{\dagger}\otimes I_d) \cdot 1
\end{align}
Taking the result of \(\rho_{A,S_1,N_1}\) from \eqref{rho_interm}, we can calculate the \(\rho_{S_1}\), by taking the partial trace over the subsystems\(A\) and \(N_1\). Using the result from \Cref{th_Tr_Ni}, and the properties of the trace, we obtain
\begin{align}
    \rho_{S_1} &= \frac{1}{d^2}\sum_{k,l=0}^{d-1}|c_{kl}|^2\cdot \text{Tr}\Big(A_{kl}\ket{\psi}_A\bra{\psi}_AA_{kl}^{\dagger}\Big)\cdot \notag \\
    &\phantom{=}\text{Tr}_{N_1}\Big((S_{kl}\otimes I_d) \ket{\Phi_d}_{S_1,N_1}\bra{\Phi_d}_{S_1,N_1}(S_{kl}^{\dagger}\otimes I_d)\Big)\cdot 1\notag\\
    &=\frac{1}{d^2}\sum_{k,l=0}^{d-1}1\cdot\text{Tr}\Big(\bra{\psi}_AA_{kl}^{\dagger}A_{kl}\ket{\psi}_A\Big)\cdot\frac{1}{d}S_{kl}S_{kl}^{\dagger}\cdot1\notag\\
    &=\frac{1}{d^2}\sum_{k,l=0}^{d-1}1\cdot1\cdot\frac{1}{d}I_d\cdot1\notag\\
    &=\frac{1}{d}I_d.
\end{align}
We proved that the encryption matrix we proposed indeed generates an encrypted state. We did not impose additional conditions on the \(S_1\) subsystem, which means that the same result will be obtained for any \(\rho_{S_i}\).

The decrypted state is obtained by applying the unitary matrix proposed for the decryption operation to the encrypted state. Thus, we obtain
\begin{align}
\label{eq_state_dec}
\ket{\psi}_{dec} &= U_{dec_d}^{(n)}\ket{\psi}_{enc} \notag\\
&=\sum_{k,l=0}^{d-1}c_{kl}^{-1}\Gamma^{\scriptscriptstyle (S_1,N_1)}O_{kl}^{\scriptscriptstyle(S_1N_1)}\ketbra{\Phi_d}{O_{kl}^{\scriptscriptstyle(S_1N_1)}}^{\scriptscriptstyle\dagger}
 \notag\\
&\phantom{=}\otimes\bigotimes_{j=2}^{n}\Big({X_d^{\scriptscriptstyle(N_j)}}^{k}{Z_d^{\scriptscriptstyle(N_j)}}^{-l}\Big)\cdot \notag\\
&\phantom{=}\cdot\frac{1}{d}\sum_{m,n=0}^{d-1}c_{mn}A_{mn}\ket{\psi}_A \bigotimes_{j=1}^{n} O_{mn}^{\scriptscriptstyle(S_j,N_j)}\ket{\Phi_d},
\end{align}
where we have used the notations \(O_{kl}^{\scriptscriptstyle(S_1,N_1)}=({X_d^{\scriptscriptstyle(S_1)}}^{k}{Z_d^{\scriptscriptstyle(S_1)}}^{l}\otimes {I_d}^{\scriptscriptstyle(N_1)})\), \(A_{kl} = {X_d^{\scriptscriptstyle(A)}}^{k}{Z_d^{\scriptscriptstyle(A)}}^{l} \), and \(\Gamma^{\scriptscriptstyle(S_1,N_1)} = \text{SWAP}^{\scriptscriptstyle(S_1,N_1)}\cdot C^{\scriptscriptstyle(S_1,N_1)}\). To show that the decryption unitary correctly retrieves in \(\ket{\psi}_{S1}\) the information originally stored in the data qudit \(\ket{\psi}_A\), we simply perform a step-by-step evolution of the state in \eqref{eq_state_dec}. Thus, we calculate and obtain the following result.
\begin{align}
    \ket{\psi}_{dec}&=\frac{1}{d}\sum_{k,l=0}^{d-1}\sum_{m,n=0}^{d-1}c_{kl}^{-1}c_{mn}\cdot\Gamma^{\scriptscriptstyle(S_1,N_1)}A_{mn}\ket{\psi}_A\otimes \notag\\
    &\phantom{==}O_{kl}^{\scriptscriptstyle(S_1,N_1)}\ket{\Phi_d}\bra{\Phi_d}{O_{kl}^{\scriptscriptstyle(S_1,N_1)}}^{\dagger}O_{mn}^{\scriptscriptstyle(S_1,N_1)}\ket{\Phi_d} \otimes\notag\\
    &\phantom{==}\bigotimes_{j=2}^{n}\Big({X_d^{\scriptscriptstyle(N_j)}}^{k}{Z_d^{\scriptscriptstyle(N_j)}}^{-l}{X_d^{\scriptscriptstyle(S_i)}}^{m}{Z_d^{\scriptscriptstyle(S_i)}}^{n}\Big)\ket{\Phi_d}\notag \\
    &\phantom{=}\text{Using the result from Lemma \ref{th_orthonormal_basis}},\notag\\
    &=\frac{1}{d}\sum_{k,l=0}^{d-1}|c_{kl}|^2\text{SWAP}^{\scriptscriptstyle(S_1,N_1)}A_{kl}\ket{\psi}_A O_{kl}^{\scriptscriptstyle(S_1,N_1)}\ket{\Phi_d} \notag \\
    &\phantom{==}\bigotimes_{j=2}^{n}\Big({X_d^{\scriptscriptstyle(N_j)}}^{k}{Z_d^{\scriptscriptstyle(N_j)}}^{-l}{X_d^{\scriptscriptstyle(S_i)}}^{k}{Z_d^{\scriptscriptstyle(S_i)}}^{l}\Big)\ket{\Phi_d} \notag\\
    &\phantom{=}\text{Using the result from Theorem \ref{th_2}},\notag\\
    &=\Gamma^{\scriptscriptstyle(S_1,N_1)}\frac{1}{d}\sum_{k,l=0}^{d-1}A_{kl}\ket{\psi}_A O_{kl}^{\scriptscriptstyle(S_1,N_1)}\ket{\Phi_d} \notag \\
    &\phantom{====}\otimes \bigotimes_{j=2}^{n}\ket{\Phi_d} \notag \\
    &\phantom{=}\text{We use the property from Lemma \ref{th_swap}}, \notag \\
    &=\text{SWAP}^{\scriptscriptstyle(S_1,N_1)}\frac{1}{\sqrt{d}}\sum_{p=0}^{d-1}\ket{p}_A\ket{p}_{S_1}\ket{\psi}_{N_1} \otimes \bigotimes_{j=2}^{n}\ket{\Phi_d}\notag \\
    &=\frac{1}{\sqrt{d}}\sum_{p=0}^{d-1}\ket{p}_A\ket{\psi}_{S_1}\ket{p}_{N_1} \otimes \bigotimes_{j=2}^{n}\ket{\Phi_d}.
\end{align}
We showed that the information stored in qudit \(A\) is transferred now to \(S_1\), as needed, without imposing any supplementary restriction on the other \(S_i\).

\section{Proof of the theorems }
\subsection{Proof of the Theorem \ref{th_2}}
\label{proof_th2}
Let \(\ket{\Phi_d}\) be  the generalized Bell state and \(X_d, Z_d\) be the generalized Pauli operators. In this Appendix, we provide the derivation for the qudit identity
\begin{equation}
    \left( {X_d^{\scriptscriptstyle(A)}}^{k_1}{Z_d^{\scriptscriptstyle(A)}}^{-k_2} \otimes {X_d^{\scriptscriptstyle(B)}}^{k_1}{Z_d^{\scriptscriptstyle(B)}}^{k_2} \right) \ket{\Phi_d} = \ket{\Phi_d},
\end{equation}
for any dimension $d \geq 2$ and $k_1, k_2 \in \{0, \dots, d-1\}$.
\begin{proof} 
We prove the theorem by fixing one of the operators, constructing the other, and showing that we obtain the same result as in the theorem statement. Specifically, we need to calculate the \(O_{A}\), and \(O_{B}\), such that 
\begin{align*}
    \Big(O_A\otimes {X_d^{k_1}}{Z_d^{k_2}}\Big)\ket{\Phi_d} &=\ket{\Phi_d}, \text{and}\\
    \Big({X_d^{k_1}}{Z_d^{-k_2}} \otimes O_{B}\Big) \ket{\Phi_d} &= \ket{\Phi_d}.
\end{align*}
For the calculation of \(O_A\), we consider that there is no operator applied on the \((A)\) subsystem. We observe the effect of the gates on the \((B)\) subsystem, and then we compute the gates that cancel that effect.
\begin{align*}
    I_d \otimes X_d^{k_1}Z_d^{k_2}\ket{\Phi_d} &= \frac{1}{\sqrt{d}}\sum_{p=0}^{d-1}\ket{p}X_d^{k_1}Z_d^{k_2}\ket{p}\\
    &= \frac{1}{\sqrt{d}}\sum_{p=0}^{d-1}\ket{p}\omega^{k_2p}\ket{p\oplus_d k_1}\\
    &\text{We perform : $m\mapsto p\oplus_dk_1$}\\
    &= \frac{1}{\sqrt{d}} \sum_{m=0}^{d-1}\omega^{(m -k_1)k_2}\ket{m\ominus_dk_1}\ket{m}
\end{align*}
Thus, we must find the operator that performs the following mapping \(
O_A (\omega^{m-k_1)k_2}\ket{m\ominus_dk_1})= \ket{m}.
\)
We can easily find that this operator is indeed \(O_A = X_d^{k_1}Z_d^{\scriptscriptstyle -k_2}\) because 
\begin{align*}
    X_d^{k_1}Z_d^{-k_2}\omega^{(m-k_1)k_2}\ket{m\ominus_dk_1} &= X_d^{k_1}\ket{m\ominus_dk_1} = \ket{m}
\end{align*}
For the calculation of \(O_B\), we consider that there is no operator applied to the \((B)\) subsystem. We observe the effect of the gates applied to the \((A)\) subsystem, and then we compute the gates that cancel that effect.
\begin{align*}
    X_d^{\scriptscriptstyle k_1}Z_d^{\scriptscriptstyle {-k_2}} \otimes I_d\ket{\Phi_d} &= \frac{1}{\sqrt{d}}\sum_{p=0}^{d-1}X_d^{\scriptscriptstyle k_1}Z_d^{\scriptscriptstyle-k_2}\ket{p}\ket{p}\\
    &= \frac{1}{\sqrt{d}}\sum_{p=0}^{d-1}\omega^{\scriptscriptstyle-k_2p}\ket{p\oplus_dk_1}\ket{p}\\
    &\text{We perform : $m\mapsto p\oplus_dk_1$}\\
    &= \frac{1}{\sqrt{d}} \sum_{m=0}^{d-1}\omega^{\scriptscriptstyle -(m-k_1)k_2}\ket{m}\ket{m\ominus_dk_1}
\end{align*}
Thus, we must find the operator that performs the following mapping 
\(
O_B (\omega^{-(m-k_1)k_2}\ket{m\ominus_dk_1})= \ket{m}.
\)
We can easily find that this operator is indeed \(O_B = X_d^{k_1}Z_d^{k_2}\) because 
\[
X_d^{k_1}Z_d^{k_2}\omega^{-(m-k_1)k_2}\ket{m-k_1} = X_d^{k_1}\ket{m-k_1} = \ket{m}.
\]
Considering that using the constructions presented above, we derive the form of the operators in the theorem statement, we have proved that the construction holds \(\forall (k_1,k_2) \in \{0,\dots,d-1\}^2\).
\end{proof}

\subsection{Proof of the Theorem \ref{th_Tr_Ni}}
\label{proof_th_Tr_Ni}
Let \(\ket{\Phi_d}\) be  the generalized Bell state, \(O_1, O_2\) be two general one-qudit operators and \(I_d^{\scriptscriptstyle(B)}\) be the identity applied on the subsystem B.
We prove the identity 
\begin{equation}
    \text{Tr}_B\Big((O_1\otimes I_d^{\scriptscriptstyle(B)})\ket{\Phi_d}_{A,B}\bra{\Phi_d}_{A,B}(O_2^{\dagger }\otimes I_d^{\scriptscriptstyle(B)})\Big) = \frac{1}{d}O_1O_2^{\dagger}.
\end{equation}
\begin{proof}
\begin{align*}
    LHS &= \frac{1}{d}\sum_{l, p_1,p_2=0}^{d-1}(I_d^{\scriptscriptstyle(A)} \bra{l}_B)(O_{1} \otimes I_d^{\scriptscriptstyle(B)})\left(\ket{p_1}_A\bra{p_2}_A\right)\cdot\\&\phantom{=======}\cdot\left(\ket{p_1}_B\bra{p_2}_B \right)(O_{2}^{\dagger} \otimes I_d^{\scriptscriptstyle(B)})(I_d^{\scriptscriptstyle(A)} \ket{l}_B)\\
    &=\frac{1}{d}\sum_{l,p_1,p_2=0}^{d-1}O_{1}\ket{p_1}_A\bra{p_2}_AO_{2}^{\dagger}\braket{l}{p_1}_B\braket{p_2}{p_1}_B\\
    &=\frac{1}{d}\sum_{l=0}^{d-1}O_{1}\ket{l}_A\bra{l}_AO_{2}^{\dagger}\\
    &=\frac{1}{d}O_1O_2^{\dagger}
\end{align*}
\end{proof}

\subsection{Proof of the  Theorem \ref{trace_gen_bell}}
\label{proof_trace_gen_bell}
Let \(\ket{\Phi_d}\) be  the generalized Bell state and \(X_d, Z_d\) be the generalized Pauli operators.
We prove the identity
\begin{equation}
    \text{Tr}\Big((X_d^{\scriptscriptstyle k}Z_d^{\scriptscriptstyle l}\otimes I_d)\ket{\Phi_d}\bra{\Phi_d}(Z_d^{\scriptscriptstyle -n}X_d^{\scriptscriptstyle -m}\otimes I_d)\Big) = \delta_{k,m}\delta_{l,n},
\end{equation}
for any \(m,n\in\{0,\dots,d-1\}\).
\begin{proof} 
\begin{align*}
    LHS &=\text{Tr}\Big((X_d^{k}Z_d^{l}\otimes I_d)\ket{\Phi_d}\bra{\Phi_d})(Z_d^{-n}X_d^{m}\otimes I_d)\Big)\\
    &\text{Using the permutation property of the trace}\\
    &=\text{Tr}\Big(\bra{\Phi_d}(Z_d^{-n}X_d^{-m}\otimes I_d)(X_d^{k}Z_d^{l}\otimes I_d)\ket{\Phi_d}\Big)\\
    &\text{Using the result from Lemma \ref{th_orthonormal_basis},}\\
    &=\text{Tr}\Big(\delta_{k,m}\delta_{l,n}\Big)=1\cdot\delta_{k,m}\delta_{l,n}
\end{align*}
\end{proof}

\bibliography{apssamp}

\end{document}